\documentclass[conference]{IEEEtran}
\makeatletter
\def\ps@headings{%
\def\@oddhead{\mbox{}\scriptsize\rightmark \hfil \thepage}%
\def\@evenhead{\scriptsize\thepage \hfil \leftmark\mbox{}}%
\def\@oddfoot{}%
\def\@evenfoot{}}
\makeatother
\usepackage{graphicx}
\usepackage{epsfig}
\usepackage{amsmath}
\usepackage{amsthm}
\usepackage{floatflt}

\newtheorem{defn}{Definition} 
 
\newtheorem{thm}{Theorem} 
\newtheorem{cor}{Corollary}

\begin{document}
\title{\textbf{CD-PHY: Physical Layer Security in Wireless Networks through Constellation Diversity}}
\author{Mohammad Iftekhar Husain, Suyash Mahant and Ramalingam Sridhar, \textit{Member, IEEE}}
\maketitle
\begin{abstract}
A common approach for introducing security at the physical
layer is to rely on the channel variations of the wireless
environment. This type of approach is not always suitable for wireless
networks where the channel remains static for most of the network
lifetime. For these scenarios, a channel independent physical layer
security measure is more appropriate which will rely on a secret known
to the sender and the receiver but not to the eavesdropper.  In this
paper, we propose CD-PHY, a physical layer security technique that
exploits the constellation diversity of wireless networks which is
independent of the channel variations. The sender and the receiver
use a custom bit sequence to constellation symbol mapping to secure the physical layer
communication which is not known a priori to the eavesdropper. Through theoretical modeling and experimental
simulation, we show that this information theoretic construct can
achieve Shannon secrecy and any brute force attack from the
eavesdropper incurs high overhead and minuscule probability of
success. Our results also show that the high bit error rate also makes
decoding practically infeasible for the eavesdropper, thus securing
the communication between the sender and receiver.
\end{abstract}
\section{Introduction}
\label{introduction}
In wireless networks, physical (PHY) layer security enables nodes to communicate securely without using resource intensive encryption mechanisms at the application layer. PHY layer security measures are resource friendly due to their information theoretic construct based on \textit{perfect secrecy}~\cite{Shannon49} in contrast with the computational hardness approaches~\cite{Woeginger03}. By introducing security at the PHY layer, communication in wireless networks can avoid the stepping stone of most attacks: \textit{eavesdropping}. In general, the broadcast nature of the the communication makes wireless networks more vulnerable to eavesdropping attacks than the wired counterpart. PHY layer security measures are able thwart such attacks to a considerable extent~\cite{EkremU09, WangYZ07}.
\par Most of the existing PHY layer security schemes are based on the variation of channel characteristics~\cite{CroftPK10, JanaPCKPK09, MathurTMYR08}. However, without highly mobile or dynamic environment which can introduce significant variation in channel characteristics, these schemes do not perform as expected~\cite{GollakotaK11}. Experimental results show that in static scenarios, these scheme mostly provide keys with very low entropy which is not desired in many cases~\cite{JanaPCKPK09}. In this paper, we propose a PHY layer security technique, CD-PHY, based on \textit{constellation diversity}, which is not dependent on channel characteristics and the performance does not vary depending on static or mobile scenario.
\par The underlying technique for CD-PHY is simple. At the physical layer, the sender and the intended receiver uses a custom constellation mapping~\cite{Takahara03} which acts as a secret key to secure the communication from an eavesdropper. In other words, a sequence of bits from the sender is converted into symbols on the constellation space based on a mapping known only to the sender and the intended receiver. Using the correct mapping, the intended receiver will be able to decode the signal and reconstruct the original message. However, the eavesdropper will not even be able to decode the signal correctly without the knowledge of constellation mapping, let alone reconstruction of the message.
\par The guarantee of security provided by CD-PHY is much stronger than just keeping the modulation type (BPSK, QPSK, and QAM\footnote{BPSK and QPSK refers to Binary and Quadrature Phase Shift Keying, respectively. QAM refers to Quadrature Amplitude Modulation. An overview of modulation schemes by Zeimer can be found at~\cite{Ziemer96}.}, for example) a secret between the sender and the receiver. Because, if the sender and receiver uses the standard constellation mapping for these modulations, an eavesdropper can use advanced machine learning techniques~\cite{Mobasseri00, Ramkumar09} to identify the modulation type and then use the standard mapping to decode the signal. In case of CD-PHY, the custom constellation mapping is known only to the sender and the receiver which is the basis of security for this information theoretic construct.
\par Our theoretical modelling, security analysis and experimental simulation reveals the following about CD-PHY:
\begin{itemize}
\item For the eavesdropper, the probability of successfully decode the symbols range from $10^{-3}$ at $10dB$ SNR\footnote{Signal-to-noise ratio.} to $0.015$ at $0dB$ SNR, which is very low (Section~\ref{theoreticalmodelling}),
\item CD-PHY achieves perfect secrecy as a cipher and has a very high unicity distance which ensures that the eavesdropper will not be able to find the correct decoding regardless of the amount of ciphertexts it collects (Section~\ref{shannonsecrecy}),
\item A brute-force key search attack on CD-PHY has complexity $\#P$ (Sharp P)\footnote{The set of the counting problems associated with the decision problems in the set NP.} which is believed to be much harder than polynomial time algorithms (Section~\ref{complexity}), and
\item Performance wise, in the presence of CD-PHY, regardless of the location, the bit error rate at the eavesdropper is always as high as $50\%$ which is equivalent to random guessing for the decoding purposes (Section~\ref{eval}).
\end{itemize}
\section{Background and Observations}
\label{backgroundandobservations}

At the physical layer, a modulation technique prepares the digital bit sequences for transmission over the analog wireless medium. A crucial part of this operation is to map the bit sequences to symbols which can be represented as points on a two dimensional complex plane called the \textit{constellation diagram}. Figure \ref{fig:16QAMCircular} shows an example constellation diagram from 16-ary Quadrature Amplitude Modulation (16QAM circular). An alternate constellation diagram is shown in Figure \ref{fig:16QAMRectangular} which is known as 16QAM rectangular. If the transmitter wishes to send a bit sequence, it sets the real (x-axis) and imaginary (y-axis) part according to the constellation diagram. Mathematically, a signal can be expressed by the following equation:
\[ s(t)=I(t).cos(2\pi f_o t)+Q(t).sin(2\pi f_o t)\]
%\begin{equation} 
%$s(t)=I(t).cos(2\pi f_o t)+Q(t).sin(2\pi f_o t)$
%\label{eq:Signal}
%\end{equation}
where I(t) and Q(t) are real and imaginary parts of the symbols from the constellation diagram and $f_o$ is the modulating frequency.
The receiver recovers the real and imaginary values after demodulation, and plots each symbol on the constellation plane. To correctly decode the original message, the receiver needs to know both the type of modulation as well as symbol to bit sequence mapping\footnote{Constellation mapping.}. 

\begin{figure*}[ht]
\begin{minipage}[b]{0.5\linewidth}
\centering
\includegraphics[scale=0.4]{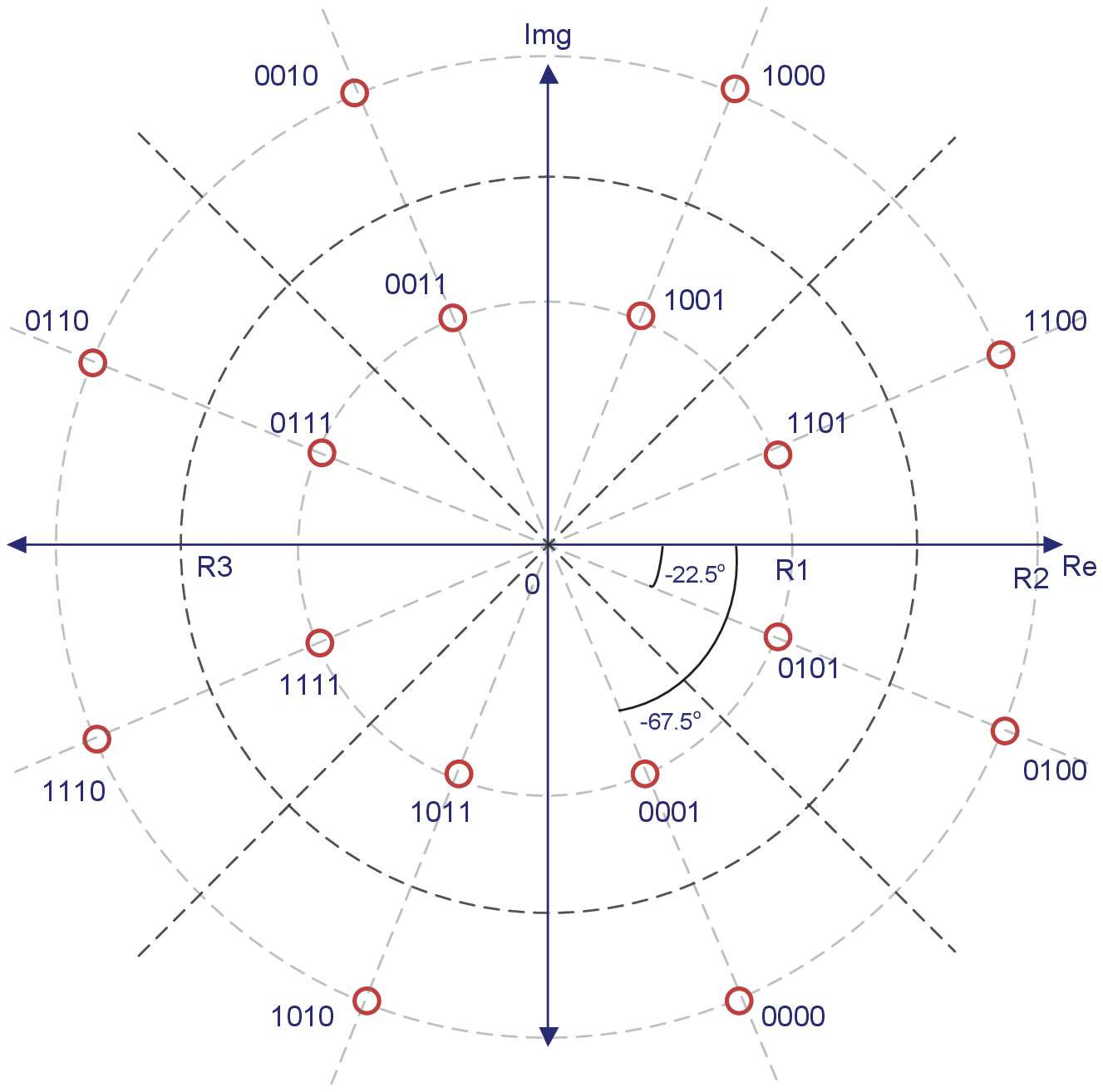}
\caption{16QAM Circular Constellation}
\label{fig:16QAMCircular}
\end{minipage}
\hspace{0.5cm}
\begin{minipage}[b]{0.5\linewidth}
\centering
\includegraphics[scale=0.4]{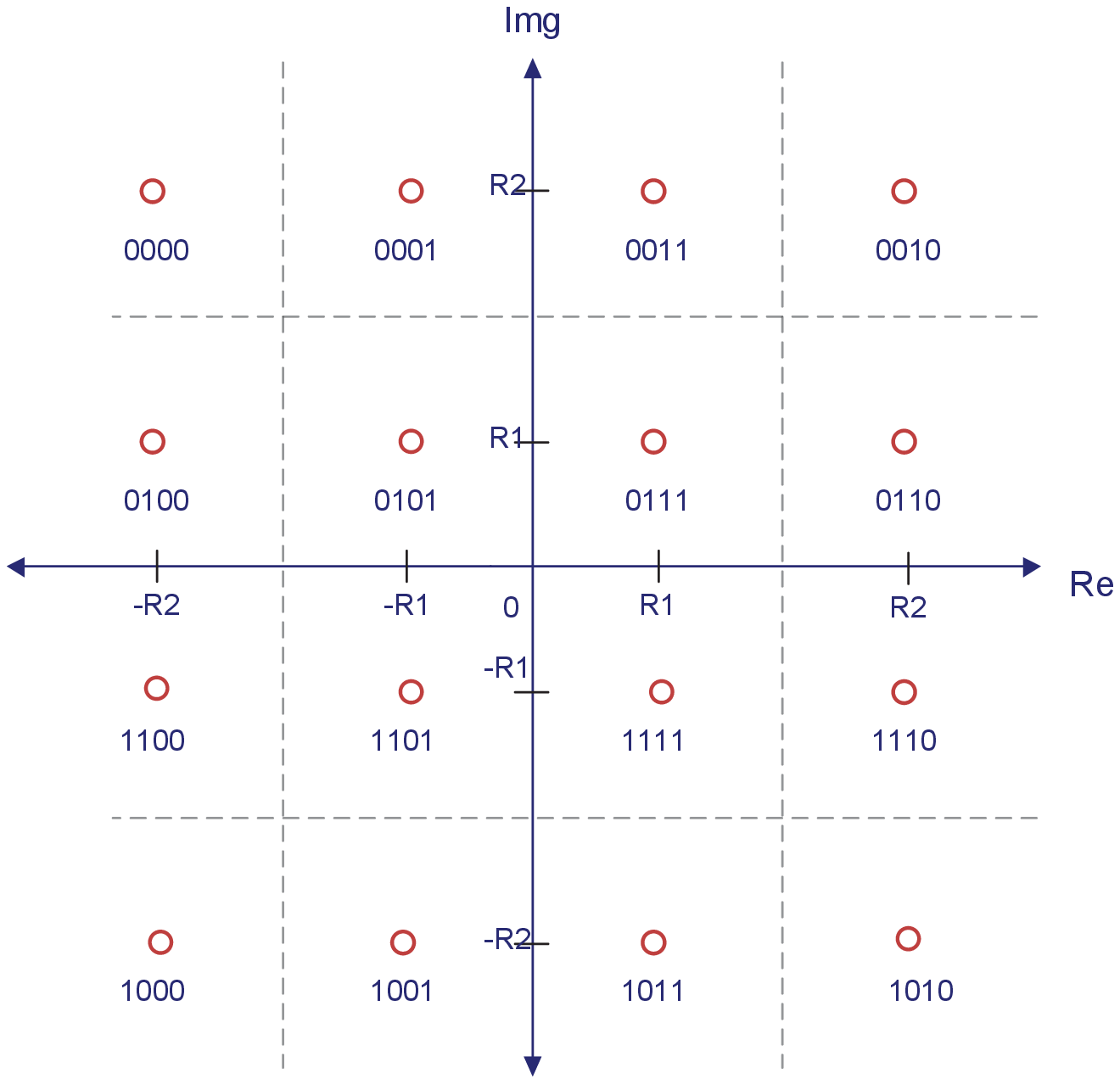}
\caption{16QAM Rectangular Constellation}
\label{fig:16QAMRectangular}
\end{minipage}
\end{figure*}

\par When only the modulation type is the secret, the eavesdropper can use machine learning based techniques~\cite{Mobasseri00, Ramkumar09} to identify the modulation type and use standard bit sequence to symbol mapping to decode the data. However, if the sender and receiver use a custom constellation mapping which is not known to the eavesdropper, the complexity of correct decoding becomes very high. For an M-ary QAM, the eavesdropper has to try all $M!$ mappings to find out the correct decoding, which is very impractical for scenarios when the value of $M\geq8$.  
\begin{figure*}[ht]
\begin{minipage}[b]{0.5\linewidth}
\centering
\includegraphics[scale=0.55]{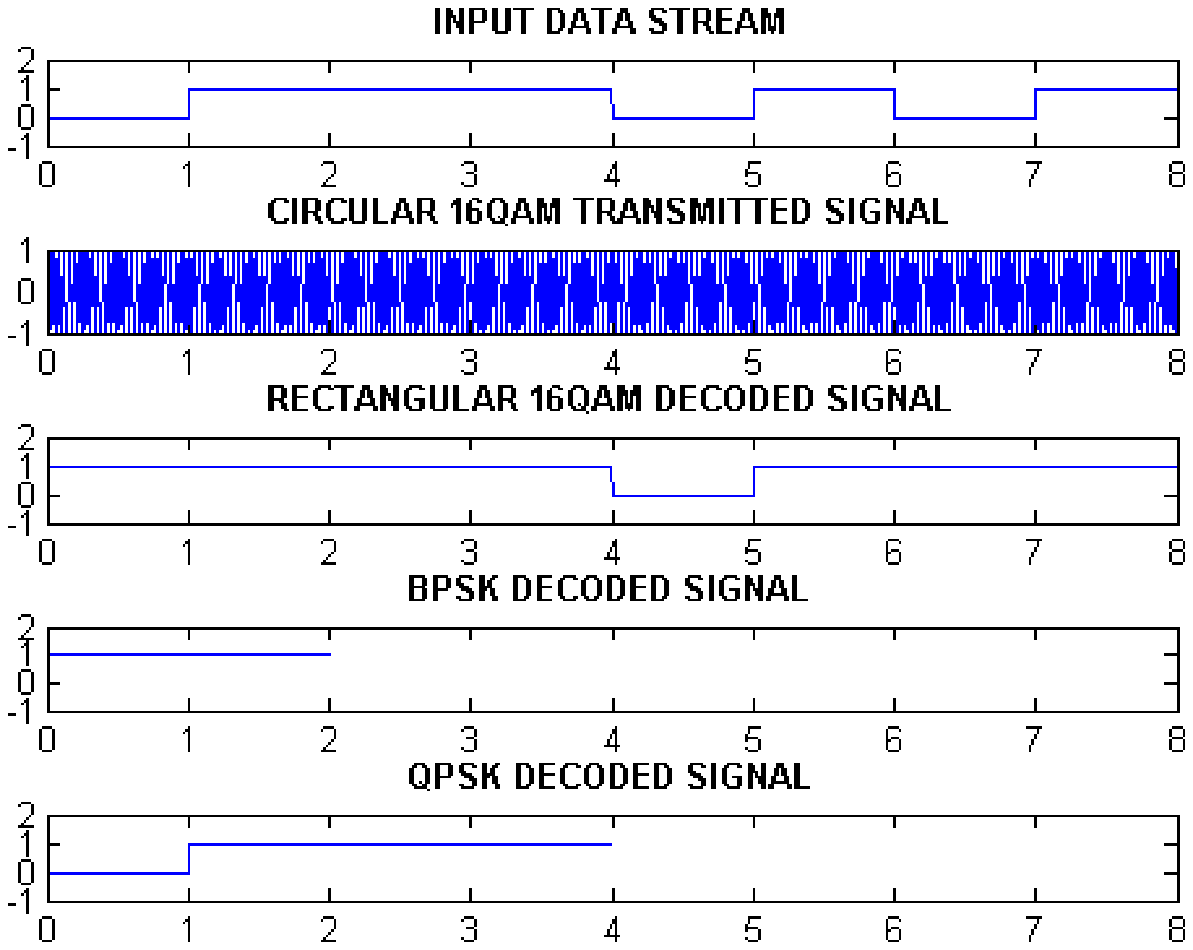}
\caption{Decoding failure when the original modulation is 16QAM circular.}
\label{fig:ModulationUnknown1}
\end{minipage}
\hspace{0.25cm}
\begin{minipage}[b]{0.5\linewidth}
\centering
\includegraphics[scale=0.55]{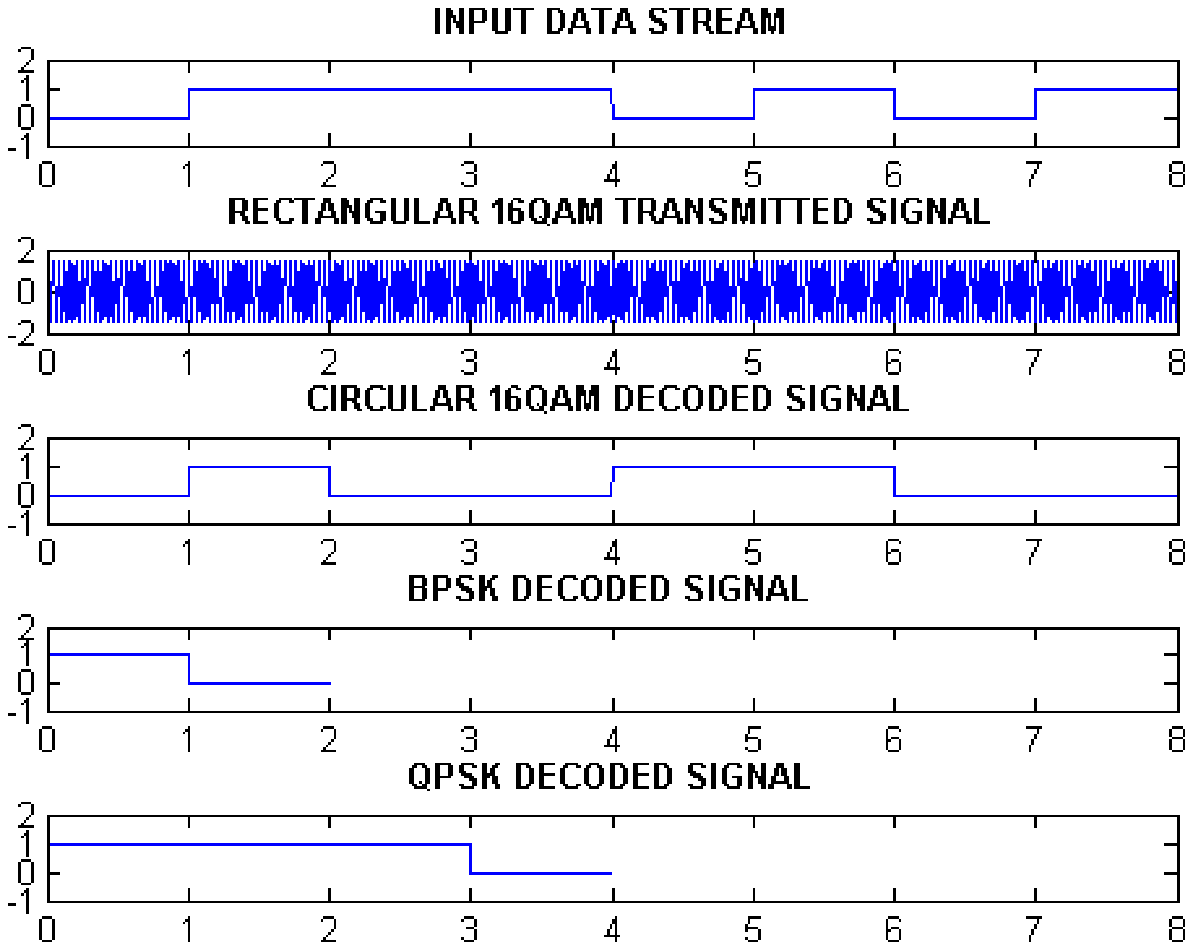}
\caption{Decoding failure when the original modulation is 16QAM rectangular.}
\label{fig:ModulationUnknown2}
\end{minipage}
\end{figure*}

%\par \textit{(ii) Modulation type known but constellation mapping unknown:}
\par Figure \ref{fig:ModulationUnknown1} shows the decoding failure when the eavesdropper tries to decode an original 16QAM circular modulated signal using different modulation types: BPSK, QPSK and 16QAM rectangular. The input data stream contained 8 bits, $01100101$. In 16QAM, each symbol consists of 4 bits. So, two symbols will be received by the eavesdropper. The QPSK receiver decodes two symbols as 4 bits and the BPSK receiver decodes it to 2 bits. Since the modulation classification was wrong, obviously the mapping will also be wrong resulting to a decoding failure. In the case of 16QAM rectangular, the receiver will correctly expand the symbols to 8 bits. However, since the constellation mapping was different\footnote{Refers to Figure \ref{fig:16QAMCircular} and \ref{fig:16QAMRectangular}.}, the final decoded data will be different from the input: $11110111$. Another decoding failure, where the original symbols belonged to 16QAM rectangular, is shown in Figure \ref{fig:ModulationUnknown2}. 
\par The intuitive design of CD-PHY is based on the above mentioned observations that without knowing the correct constellation mapping, it is not practically feasible for an eavesdropper to correctly decode the original message even though it might have the knowledge of modulation type.

\section{Adversarial Model}
We assume that the adversary (eavesdropper) is able to detect and will try to decode the communication between the sender and receiver. It can be either mobile or static. An adversary can also measure the channel parameters. It can exploit some machine learning techniques to identify the modulation type of the wireless communication, but it does not have prior knowledge of the constellation mapping between the sender and intended receiver.
\par We also assume the eavesdropper's computation and communication capability as powerful as the sender and receiver. The adversary can try to handle the original signal as noise or try interference cancellation and joint decoding. Finally, we assume that the adversary is passive and has no intention to launch active attacks such as a man-in-the-middle attack. This is a common assumption among most of the practical wireless security schemes \cite{GollakotaK11}. 
\section{Theoretical Modelling}
\label{theoreticalmodelling}
In this section, we derive the probability of an eavesdropper to correctly decode the message in the presence of gaussian noise when it knows the modulation type but does not know the constellation mapping. A very intuitive example of this case is the interaction between 16QAM circular and rectangular modulations discussed in Section \ref{backgroundandobservations}. We use this example to derive the probability measure of correct decoding when the sender modulation is 16QAM circular and eavesdropper modulation is 16QAM rectangular. 
\par As discussed in Section \ref{backgroundandobservations}, each QAM symbol has a real and imaginary value associated with it in the constellation space. Mathematically, for an M-ary QAM, these real and imaginary values can range $\pm a,\hdots,\pm (2m-1)a$, where $m=\log_2 M$, $a^2=1.5E_s/(M-1)$ with $E_s$ being the symbol energy \cite{digicom2}. 
%For the 16QAM circular modulation in Figure \ref{fig:16QAMCircular}, the points on the constellation are ${(\pm+\pm j2),(\pm2+\pm j4),(\pm4+\pm j2),(\pm4+\pm j4)}$. Similarly, for the 16QAM rectangular modulation in Figure \ref{fig:16QAMRectangular}, the points take the values of ${(\pm1+\pm j1),(\pm1+\pm j3), (\pm3+\pm j1),(\pm3+\pm j3)}$. 
Table \ref{tab:ConstellationMapping} shows the bit sequence to constellation symbol mapping in 16QAM circular and 16QAM rectangular Scheme
These values are further factored by $a=\sqrt{E_s/10}$ to normalize the average symbol energy to 1. 
\begin{table*}[htbp]
  \centering
  \caption{Bit sequence to constellation symbol mapping in 16QAM Circular and 16QAM Rectangular Scheme}
    \begin{tabular}{cccccc}
    \hline
    \textbf{Bit Sequence} & \textbf{16QAM Circular} & \textbf{16QAM Rectangular} & \textbf{Bit Sequence} & \textbf{16QAM Circular} & \textbf{16QAM Rectangular} \\
    \hline
    \hline
    $0000$ & $1.53-3.69j$ & $-3+3j$ & $1000$ & $1.53+3.69j$ & $-3-3j$ \\
    $0001$ & $.76-1.84j$ & $-1+3j$ & $1001$ & $.76+1.84j$ & $-1-3j$ \\
    $0010$ & $-1.53+3.69j$ & $3+3j$ & $1010$ & $-1.53-3.69j$ & $3-3j$ \\
    $0011$ & $-.76+1.84j$ & $1+3j$ & $1011$ & $-.76-1.84j$ & $1-3j$ \\
    $0100$ & $3.69-1.53j$ & $-3+j$ & $1100$ & $3.69+1.53j$ & $-3-j$ \\
    $0101$ & $1.84-.76j$ & $-1+j$ & $1101$ & $1.84+.76j$ & $-1-j$ \\
    $0110$ & $-3.69+1.53j$ & $3+j$ & $1110$ & $-3.69-1.53j$ & $3-j$ \\
    $0111$ & $-1.84+.76j$ & $1+j$ & $1111$ & $-1.84-.76j$ & $1-j$ \\
    \hline
    \end{tabular}%
  \label{tab:ConstellationMapping}%
\end{table*}%

\par The decision variable for demodulation in the presence of \textit{additive white gaussian noise} can be obtained as
\begin{equation}
	Y \approx X + n 
\end{equation}
where the noise term $n(t)$ is assumed with power spectral density $\frac{N_o}{2}$, zero mean and variance of $\sigma^2 = N_o$. Thus, the decision variable $Y$ is a complex gaussian with a complex mean $X$ and variance $\sigma^2 = N_o$. In other words, $Y$ has a two dimensional gaussian distribution in complex plane. So, the real and imaginary parts of $Y$ can be separated as independent gaussian variables as $Y_R$ and $Y_I$ with means at $Re(X)$ and $Im(X)$.
%\begin{equation}
\[ Y_R = Re(Y) = Re(X) + n_R = X_R + n_R \]
\[ Y_I = Im(Y) = Im(X) + n_I = X_I + n_I \]
%\end{equation}
where $n_R$ and $n_I$ are the components of noise along real and imaginary axes with a mean zero and variance $\sigma_R^2 = \sigma_I^2 = \frac{N_o}{2}$. Now, the probability density function of $Y_R$ can be expressed as: 
\begin{eqnarray}
f(Y_R) &=& \frac{1}{\sqrt{2\pi \sigma_R^2}}\  exp - \{ \frac{(Y_R-X_R)^2}{2\sigma_R^2}\} \nonumber \\
 &=& \frac{1}{\sqrt{\pi N_o}}\  exp - \{ \frac{(Y_R-X_R)^2}{N_o}\}
\end{eqnarray}
Similarly, the probability density function of $Y_I$ can also be expressed as:
\begin{equation}
	f(Y_I) = \frac{1}{\sqrt{\pi N_o}}\  exp - \{ \frac{(Y_I-X_I)^2}{N_o}\}
\end{equation}
%Thus to find the probability that a Malicious node correctly decodes a secure symbol we need to consider 16-QAM rectangular and Circular constellations. Fig.~\ref{fig:16QAMRectangular} \&~\ref{fig:16QAMCircular} show rectangular and circular constellation diagrams with soft-decision boundaries indicated by dashed lines.\cite{rect2}\cite{circ3}

Now, to calculate the probability of the successful decoding at the eavesdropper with 16QAM rectangular scheme when the original symbols were transmitted in 16QAM circular scheme, we first need to consider the probabilities at individual symbol level. These probabilities are then aggregated using the symmetry and mutual independence of the symbols. In the following derivations, $S_i^r$ denotes a symbol $S_i$ in 16QAM rectangular scheme, $S_i^c$ represents a symbol $S_i$ in 16QAM circular scheme and  four symbols are chosen from the constellation diagram in such a way that symmetrically they represent all sixteen points of a QAM scheme.
\subsection{Decoding of symbol $0000$} 
First, we consider $S_0^c = 0000$ being transmitted. From Table \ref{tab:ConstellationMapping}, the real and imaginary parts of $0000$ are 
	\[ X_R = 1.53 \sqrt{\frac{E_s}{10}}\ \ \& \ \ X_I = -3.69 \sqrt{\frac{E_s}{10}}   \]
	The received symbol $Y$ has a complex gaussian distribution as discussed earlier with the mean at $X_R+jX_I$. Now, the probability that the symbol $Y$ can be correctly decoded by the eavesdropper using 16QAM rectangular decoder can be found based on the decision space for $S_0^r=0000$ in 16QAM rectangular scheme. Formally, the probability that decoded symbol is $S_0^r$ given $S_0^c$ was transmitted is:
\begin{eqnarray}
	\nonumber
	\lefteqn{\scalebox{0.95}{$P(Y = S_0^r | S_0^c) = $}}\\
  \nonumber
	& \scalebox{0.95}{$P\!\left(\!-\infty < Y_R < -2\sqrt{\frac{E_s}{10}}\!\right) P\!\left(\!2\sqrt{\frac{E_s}{10}} < Y_I < \infty\!\right)$}
	\end{eqnarray}
	\[P(Y = S_0^r | S_0^c) = \int_{-\infty}^{-2\sqrt{\frac{E_s}{10}}} f(Y_R)dY_R\times \int_{2\sqrt{\frac{E_s}{10}}}^{\infty} f(Y_I)dY_I \]
	\begin{eqnarray}
	\nonumber
	\lefteqn{ P(Y = S_0^r | S_0^c) = }\\
	\nonumber
	& \frac{1}{\sqrt{\pi N_o}} \int_{-\infty}^{-2\sqrt{\frac{E_s}{10}}} exp-\{\frac{(Y_R-1.53\sqrt{\frac{E_s}{10}})^2}{N_o}\} dY_R\\
	\nonumber
	& \times \frac{1}{\sqrt{\pi N_o}} \int_{2\sqrt{\frac{E_s}{10}}}^{\infty} exp-\{\frac{(Y_I-(-3.69\sqrt{\frac{E_s}{10}}))^2}{N_o}\} dY_I \\
	\nonumber
	\end{eqnarray}
	Using the simplification of above integrals,
	\begin{equation}
	\begin{split}
	P(Y = S_0^r | S_0^c) = \frac{1}{\sqrt{\pi}} \int_{3.53 \sqrt{\frac{E_s}{10 N_o}}}^{\infty} exp\{ -t^2\}dt \\
	\times \frac{1}{\sqrt{\pi}} \int_{5.69 \sqrt{\frac{E_s}{10 N_o}}}^{\infty} exp\{ -z^2\}dz 
	\nonumber
	\end{split}
	\end{equation}
	\begin{eqnarray}
	\nonumber
	\lefteqn{P(Y = S_0^r | S_0^c)}\\
	\nonumber
	& = &\scalebox{0.95}{$\frac{1}{2} erfc\left(\!3.53\sqrt{\frac{E_s}{10 N_o}}\!\right)\frac{1}{2} erfc\left(\!5.69\sqrt{\frac{E_s}{10 N_o}}\!\right)$}\\
	\nonumber
	& = &\scalebox{0.95}{$\frac{1}{4} erfc\left(\!3.53\sqrt{\frac{E_s}{10 N_o}}\!\right) erfc\left(\!5.69\sqrt{\frac{E_s}{10 N_o}}\!\right)$}\\
  \end{eqnarray}
 Here, $erfc()$ is the \textit{complementary error function}. 
  \subsection{Decoding of symbol $0100$}
	 Now, we consider the symbol $S_1^c = 0100$ being transmitted. Similar to the previous example, 
	\[  X_R = 3.69 \sqrt{\frac{E_s}{10}}\ \ \& \ \ X_I = -1.53 \sqrt{\frac{E_s}{10}}   \]
	So, the probability that the eavesdropper correctly decodes the symbol $0100$ is: 
	\begin{eqnarray}
	\nonumber
	\lefteqn{\scalebox{0.95}{$ P(Y = S_1^r| S_1^c) = $}}\\
	%\nonumber
	& \scalebox{0.95}{$P\left(-\infty < Y_R < -2\sqrt{\frac{E_s}{10}}\right) P\left(0 < Y_I < 2\sqrt{\frac{E_s}{10}}\right)$}
	\label{eq:0100}
	\end{eqnarray}
	Now, the left part of the right hand side of Equation \ref{eq:0100} gives us the following:
	% $P\left(-\infty < Y_R < -2\sqrt{\frac{E_s}{10}}\right)$
	\begin{eqnarray}
	\nonumber
	\lefteqn{\scalebox{0.9}{$P\left(-\infty < Y_R < -2\sqrt{\frac{E_s}{10}}\right)$} }\\
	\nonumber
	& =\! &\scalebox{1}{$\frac{1}{\sqrt{\pi N_o}} \int_{-\infty}^{-2\sqrt{\frac{E_s}{10}}}\!exp\!-\!\{\frac{(Y_R-3.69\sqrt{\frac{E_s}{10}})^2}{N_0}\}dY_R$}\\
	%\nonumber
	& =\! &\scalebox{0.95}{$\frac{1}{2}erfc\left(5.69\sqrt{\frac{E_s}{10N_o}}\right)$}
	\label{eq:0100left}
	\end{eqnarray}
	Next, the right part yields the following:  
	%$P\left(0 < Y_I < 2\sqrt{\frac{E_s}{10}}\right)$
	\[ \scalebox{0.95}{$P\left(0 < Y_I < 2\sqrt{\frac{E_s}{10}}\right) = 1 - P\left(Y_I<0 , Y_I>2\sqrt{\frac{E_s}{10}}\right) $}\]
	\begin{eqnarray}
	\nonumber
	\lefteqn{\scalebox{0.95}{$ P\left(0 < Y_I < 2\sqrt{\frac{E_s}{10}}\right) = 1 - $}}\\
	\nonumber
	& [\frac{1}{\sqrt{\pi N_o}}\int_{-\infty}^{0}\!exp\!-\!\{\frac{(Y_I-(-1.53\sqrt{\frac{E_s}{10}})))^2}{N_o}\}dY_I\\
	\nonumber
	& + \frac{1}{\sqrt{\pi N_o}}\int_{2\sqrt{\frac{E_s}{10}}}^{\infty}\!exp\!-\!\{\frac{(Y_I-(-1.53\sqrt{\frac{E_s}{10}})))^2}{N_o}\}dY_I ]
	\end{eqnarray}
	\begin{equation}
	\begin{split}
	\scalebox{0.9}{$P\left(0 < Y_I < 2\sqrt{\frac{E_s}{10}}\right) = 1 - \frac{1}{2}erfc\left(1.53\sqrt{\frac{E_s}{10N_o}}\right)$} \\
	\scalebox{0.9}{$- \frac{1}{2}erfc\left(3.53\sqrt{\frac{E_s}{10N_o}}\right)$} %\nonumber
	\end{split}
	\label{eq:0100right}
	\end{equation}
	Using Equation \ref{eq:0100left},\ref{eq:0100right} on Equation \ref{eq:0100}, we have the following:
	\begin{eqnarray}
	\nonumber
 	\lefteqn{P(Y = S_1^r|S_1^c) = \scalebox{0.95}{$\frac{1}{2}erfc\left(5.69\sqrt{\frac{E_s}{10N_o}}\right) \times $}}\\
 	\nonumber
 	& \left[1 - \frac{1}{2}erfc\left(1.53\sqrt{\frac{E_s}{10N_o}}\right) - \frac{1}{2}erfc\left(3.53\sqrt{\frac{E_s}{10N_o}}\right)\right] 
 	\end{eqnarray}
 	\begin{eqnarray}
 	\nonumber
  \lefteqn{P(Y = S_1^r|S_1^c) = \scalebox{0.95}{$\frac{1}{2}erfc\left(5.69\sqrt{\frac{E_s}{10N_o}}\right) $}}\\
  \nonumber
  & - \scalebox{0.95}{$\frac{1}{4}erfc\left(5.69\sqrt{\frac{E_s}{10N_o}}\right) \ erfc\left(1.53\sqrt{\frac{E_s}{10N_o}}\right)$} \\	
  & - \scalebox{0.95}{$\frac{1}{4}erfc\left(5.69\sqrt{\frac{E_s}{10N_o}}\right) erfc\left(3.53\sqrt{\frac{E_s}{10N_o}})\right) $}
  \end{eqnarray}
 
\subsection{Decoding of symbol $0101$}
 Now, we consider, $S_2^c = 0101$ is being transmitted. In this case: 
	\[ X_R = 1.84\sqrt{\frac{E_s}{10}}\ \ \& \ \ X_I = -0.76\sqrt{\frac{E_s}{10}}\]
So, the	probability that the eavesdropper correctly decodes the symbol is:
	\begin{eqnarray}
	\nonumber
	\lefteqn{ P(Y = S_2^r|S_2^c) = }\\
	%\nonumber
	& \scalebox{0.95}{$P\left(\!-2\sqrt{\frac{E_s}{10}} < Y_R < 0\!\right) P\left(\!0 < Y_I < 2\sqrt{\frac{E_s}{10}}\!\right)$}
	\label{eq:0101}
	\end{eqnarray}
	We first consider the left part of the right hand side of Equation \ref{eq:0101}: 
	%$P\left(-2\sqrt{\frac{E_s}{10}} < Y_R < 0\right)$
	\[ \scalebox{0.95}{$P\left(\!-2\sqrt{\frac{E_s}{10}} < Y_R < 0\!\right) = 1 - P\left(\!Y_R>0 , Y_R < -2\sqrt{\frac{E_s}{10}}\!\right)$}\]
	\begin{eqnarray}
	\nonumber
	\lefteqn{\scalebox{0.95}{$P\left(-2\sqrt{\frac{E_s}{10}} < Y_R < 0\!\right) = 1 - $}}\\
	\nonumber
	& [\frac{1}{\sqrt{\pi N_o}}\int_{0}^{\infty}\! exp\!-\!\{\frac{(Y_R-1.84\sqrt{\frac{E_s}{10}})^2}{N_o}\}dY_R \times \\
	\nonumber
	& \frac{1}{\sqrt{\pi N_o}}\int_{-\infty}^{-2\sqrt{\frac{E_s}{10}}}exp-\{\frac{(Y_R-1.84\sqrt{\frac{E_s}{10}})^2}{N_o}\}dY_R] 
	\end{eqnarray}
	\begin{eqnarray}
	\nonumber
  \lefteqn{\scalebox{0.95}{$P\left(\!-2\sqrt{\frac{E_s}{10}} < Y_R < 0\!\right) = 1 - $}}\\
  %\nonumber
  & \frac{1}{2}erfc\left(\!3.84\sqrt{\frac{E_s}{10N_o}}\!\right) - \frac{1}{2}erfc\left(\!-1.84\sqrt{\frac{E_s}{10N_o}}\!\right)
	\label{eq:0101left}
	\end{eqnarray}
	Similarly, we consider the right part of the right hand side of Equation \ref{eq:0101}:
	%$P\left(0 < Y_I < 2\sqrt{\frac{E_s}{10}}\right)$
	\[ \scalebox{0.95}{$P\!\left(\!0 < Y_I < 2\sqrt{\frac{E_s}{10}}\!\right) = 1 - P\!\left(\!Y_I<0 , Y_I>2\sqrt{\frac{E_s}{10}}\!\right)$} \]
	\begin{eqnarray}
	\nonumber
  \lefteqn{\scalebox{0.95}{$P\left(\!0 < Y_I < 2\sqrt{\frac{E_s}{10}}\!\right) = 1 - $}}\\
  \nonumber
  & [\frac{1}{\sqrt{\pi N_o}}\int_{-\infty}^{0}\!exp\!-\!\{\frac{(Y_I-(-0.76\sqrt{\frac{E_s}{10}}))^2}{N_o}\}dY_I \\
  \nonumber
	& + \frac{1}{\sqrt{\pi N_o}}\int_{2\sqrt{\frac{E_s}{10}}}^{\infty}\!exp\!-\!\{\frac{(Y_I-(-0.76\sqrt{\frac{E_s}{10}}))^2}{N_1}\}dY_I\ ]
	\end{eqnarray}
	\begin{equation}
	\begin{split}
  \scalebox{0.9}{$P\!\left(\!0 < Y_I < 2\sqrt{\frac{E_s}{10}}\!\right) = 1 - \frac{1}{2}erfc\left(\!-0.76\sqrt{\frac{E_s}{10N_o}}\!\right) $}\\
  \scalebox{0.9}{$- \frac{1}{2}erfc\left(\!2.76\sqrt{\frac{E_s}{10N_o}}\!\right)$}
  \end{split}
  \label{eq:0101right}
  \end{equation}
	Thus, combining Equation \ref{eq:0101left}, \ref{eq:0101right}, we have:
	\begin{eqnarray}
	\nonumber
	\lefteqn{\scalebox{0.95}{$P(Y = S_2^r|_2^c) = $}}\\
	\nonumber
	& \scalebox{0.9}{$\left[1\!-\!\frac{1}{2}erfc\!\left(\!3.84\sqrt{\frac{E_s}{10N_o}}\!\right)\!-\!\frac{1}{2}erfc\!\left(\!-1.84\sqrt{\frac{E_s}{10N_o}}\!\right)\right] $}\\
	\nonumber
	& \scalebox{0.9}{$\times\! \left[1\!-\!\frac{1}{2}erfc\!\left(\!-0.76\sqrt{\frac{E_s}{10N_o}}\!\right)\!-\!\frac{1}{2}erfc\!\left(\!2.76\sqrt{\frac{E_s}{10N_o}}\!\right)\right]$}\\
  \end{eqnarray}
  \subsection{Decoding of symbol $0001$}
  Finally, we consider symbol $S_3^c = 0001$ being transmitted. In this case: 
 \[ X_R\! =\! 0.76\sqrt{\frac{E_s}{10}}\ \ \& \ \ X_I\! =\! -1.84\sqrt{\frac{E_s}{10}}\]
 So, the probability that eavesdropper correctly decodes symbol $0001$ is:
 \begin{equation}
 %\nonumber                          
\scalebox{0.87}{$P\!\left(\!Y\!=\!S_3^r|S_3^c\right)\!=\!P\!\left(\!\!-2\sqrt{\frac{E_s}{10}}\!\!<\!\!Y_R<0\!\right)\!P\!\left(\!2\sqrt{\frac{E_s}{10}}\!<\!Y_I\!<\!\infty\!\right)$}
 \label{eq:0001}
 \end{equation}
Considering the left part of the right hand side of Equation \ref{eq:0001}:
 \begin{equation}
 \nonumber
 \scalebox{0.9}{$P\!\left(\!-2\sqrt{\frac{E_s}{10}}\!<\!Y_R\!<\!0\right)\! =\! 1\! -\! P\!\left(\!Y_R\!<\!-2\sqrt{\frac{E_s}{10}},Y_R\!>\!0\!\right)$}
 \end{equation}
 \begin{eqnarray}
 \nonumber
 \lefteqn{\scalebox{0.9}{$P\!\left(-2\sqrt{\frac{E_s}{10}}\!<\!Y_R\!<\!0\right) =\  1\ -\  $}} \\
 \nonumber
 & \scalebox{1}{$[\frac{1}{\sqrt{\pi N_o}}\int_{-\infty}^{-2\sqrt{\frac{E_s}{10}}}\!exp\!-\!\{\frac{(Y_R-0.76\sqrt{\frac{E_s}{10}})^2}{N_o}\}dY_R$} \\
 \nonumber
 & \scalebox{1}{$+ \frac{1}{\sqrt{\pi N_o}}\int_{0}^{\infty}\!exp\!-\!\{\frac{(Y_R-0.76\sqrt{\frac{E_s}{10}})^2}{N_o}\}dY_R]$}
 \end{eqnarray}
 \begin{equation}
 \begin{split}
 %\nonumber
 \scalebox{0.9}{$P\!\left(\!-2\sqrt{\frac{E_s}{10}}\!<\!Y_R\!<\!0\!\right)\! =\! 1 - \frac{1}{2}erfc\!\left(\!2.76\sqrt{\frac{E_s}{10N_o}}\!\right)$}\\
 \scalebox{0.9}{$- \frac{1}{2}erfc\!\left(\!-0.76\sqrt{\frac{E_s}{10N_o}}\!\right)$}
 \end{split}
 \label{eq:0001left}
 \end{equation}
Similarly, the right part yields:
 \begin{eqnarray}
 \nonumber
 \lefteqn{\scalebox{0.88}{$P\left(2\sqrt{\frac{E_s}{10}}\!<\!Y_I\!<\!\infty\right) = $}}\\
 \nonumber
 & \frac{1}{\sqrt{\pi N_o}}\! \int_{2\sqrt{\frac{E_s}{10}}}^{\infty}\!exp\!-\!\{\frac{\!(\!Y_I\!-(-1.84\sqrt{\frac{E_s}{10}}))^2}{N_o}\}dY_I
 \end{eqnarray}
 \begin{equation}
 %\nonumber
 \scalebox{0.9}{$P\left(2\sqrt{\frac{E_s}{10}}\!<\! Y_I\!< \!\infty\right) = \frac{1}{2}erfc\left(3.84\sqrt{\frac{E_s}{10N_o}}\right)$}
 \label{eq:0001right}
 \end{equation}
 By combining the outcomes of Equation \ref{eq:0001left}, \ref{eq:0001right}, we get the following:
 \begin{eqnarray}
 \nonumber
 \lefteqn{\scalebox{0.9}{$P(Y = S_3^r|S_3^c) = \frac{1}{2}erfc\left(3.84\sqrt{\frac{E_s}{10N_o}}\right)$}} \\
 \nonumber
 & \scalebox{0.9}{$- \frac{1}{4}erfc\left(3.84\sqrt{\frac{E_s}{10N_o}}\right) erfc\left(2.76\sqrt{\frac{E_s}{10N_o}}\right)$}\\
 & \scalebox{0.9}{$- \frac{1}{4}erfc\left(3.84\sqrt{\frac{E_s}{10N_o}}\right) erfc\left(-0.76\sqrt{\frac{E_s}{10N_o}}\right)$}	
\end{eqnarray}
%\end{enumerate}
As mentioned earlier, based on the symmetry of QAM constellation diagrams, other symbols will also have probabilities equal to one of the following symbols: $S_0,S_1,S_2$ or $S_3$. Assuming all symbols have equal probability of being generated and transmitted i.e. $P(S_k)=1/16$ where $(k=0 \ldots 15)$, the total probability $P(C)$ that the data transmitted by 16QAM circular transmitter and correctly decoded by 16QAM rectangular eavesdropper is:
 \[ P(C) = P(S_k)\times4\times[P(Y=S_0^r|S_0^c)+P(Y=S_1^r|S_1^c)\]
\[\ \ \  +P(Y=S_2^r|S_2^c)+P(Y=S_3^r|S_3^c)] \] 
\begin{equation}
\begin{split}
P(C) = \frac{1}{4}[-\frac{1}{4}erfc\!\left(\!5.69\sqrt{\frac{E_s}{10N_o}}\!\right)\!erfc\!\left(\!1.53\sqrt{\frac{E_s}{10N_o}}\!\right)\\
+ \frac{1}{4}erfc\!\left(\!-1.84\sqrt{\frac{E_s}{10N_o}}\!\right)\!erfc\!\left(\!2.76\sqrt{\frac{E_s}{10N_o}}\!\right)\\
+ \frac{1}{4}erfc\!\left(\!-1.84\sqrt{\frac{E_s}{10N_o}}\!\right)\!erfc\!\left(\!-0.76\sqrt{\frac{E_s}{10N_o}}\!\right)\\
- \frac{1}{2}erfc\!\left(\!-0.76\sqrt{\frac{E_s}{10N_o}}\!\right)\!-\!\frac{1}{2}erfc\!\left(\!2.76\sqrt{\frac{E_s}{10N_o}}\!\right)\\
+ \frac{1}{2}erfc\!\left(\!5.69\sqrt{\frac{E_s}{10N_o}}\!\right)\!-\!\frac{1}{2}erfc\!\left(\!-1.84\sqrt{\frac{E_s}{10N_o}}\!\right)\!+\!1]
\end{split}
\label{eq:finalprob}
\end{equation}
Here, $N_o$ is the power spectral density of the noise and $E_s$ is the symbol energy of the signal. So, the term $E_s/N_o$ is a representative of the SNR at the receiver. Since Equation \ref{eq:finalprob} contains $erfc()$ function, as we increase the value of SNR in the $erfc()$ function, the probability will decrease. So, the probability of correct decoding is adversely affected by the SNR of the wireless medium at receivers. This theoretical fact is illustrated further in Figure \ref{fig:prob1}. 
\begin{figure}[ht]
	\centering	
	\includegraphics[scale=0.5]{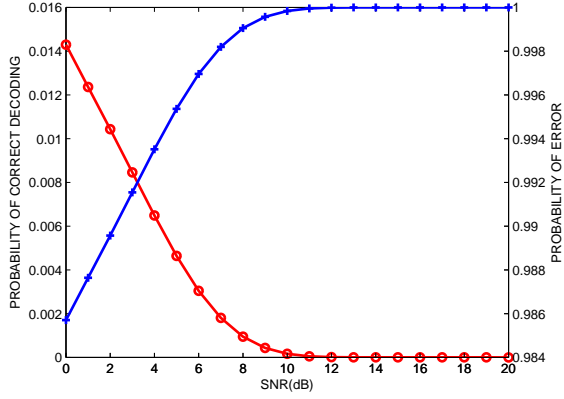}
	\caption{Probability that the eavesdropper decodes correctly and incorrectly at different signal-to-noise ratio.}
	\label{fig:prob1}
\end{figure}
The line with circles refers to the probability of correct decoding and the line with crosses refers to the probability of error. At $0dB$ SNR,  the probability of error for the eavesdropper is $0.002$. At SNR values above $20dB$, the probability of error is nearly $1$ which makes the decoding almost infeasible in practice. In comparison, for an intended receiver with 16QAM circular scheme, the probability of error at $0dB$ SNR is around $1$, and $0$ for a SNR of $20dB$~\cite{digicom2}.  
%Probability of false detection can also be found as $P(Error) = 1 - P(C)$.\\
%calculate the probability of a intended receiver and compare it with this one. 

\section{Security analysis}
In this section, we analyze CD-PHY in terms of information theoretic security, security by complexity and resistance to potential modulation classification schemes such as Automatic Modulation Classification (AMC) \cite{Ramkumar09} and Digital Modulation Classification (DMC) \cite{Mobasseri00}.
\par The basis of information theoretic security is the fact that the bit sequence to constellation symbol mapping is known only to the sender and receiver(s). The eavesdropper does not have any a prior knowledge of the mapping. In the subsequent section, by applying Shannon's secrecy model (Figure \ref{fig:shannon}) to CD-PHY, we show that it can in deed achieve information theoretic security. In addition, any decoding attempt on the eavesdropper side incurs high complexity as it blindly tries to find the mapping. Finally, we show how CD-PHY thwarts the classification attempts by AMC and DMC.
\subsection{Information theoretic security} 
\label{shannonsecrecy}
\begin{figure}[ht]
	\centering
		\includegraphics[width=0.40\textwidth]{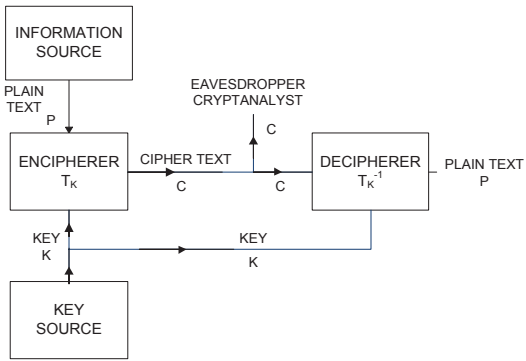}
	\caption{Shannon's Secrecy Model}
	\label{fig:shannon}
\end{figure}
In CD-PHY, the act of finding the correct mapping from the constellation points to bit sequences is essentially a deciphering operation for the eavesdropper. Here, the transmitted bit sequences are plaintext $P$, signal received by the eavesdropper is the ciphertext $C$, the mapping is the key $K$. For an M-ary QAM, the plaintext can have $M$ symbols, each of which are $\log_2 M$ bits. The key, mapping of bit sequences to constellation points, has $M!$ variations. Now, we define \textit{perfect secrecy} and \textit{unicity distance} which is due to Shannon \cite{Shannon49}.
\begin{defn} A cipher achieves perfect secrecy, if without knowing the secret key, the plaintext $P$ is independent of the ciphertext $C$, formally:
\begin{equation}
prob(\textbf{P}=P|\textbf{C}=E_K(P))=prob(\textbf{P}=P)
\label{eq:perfectsecrecy1}
\end{equation}
Equivalently, 
\begin{equation}
prob(\textbf{C}=C|\textbf{P}=E_K^{-1}(C))=prob(\textbf{C}=C)
\label{eq:perfectsecrecy2}
\end{equation}
\label{defn:perfectsecrecy}
\end{defn}
\begin{figure}[ht]
	\centering
		\includegraphics[width=0.25\textwidth]{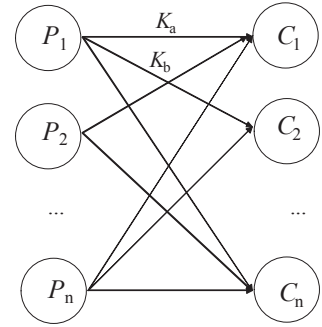}
	\caption{An illustration of plaintext to ciphertext mapping.}
	\label{fig:perfectcipher}
\end{figure}
\begin{defn}
Unicity distance of a cipher is the minimum amount of ciphertext needed for brute-force attack to succeed. Formally:
\begin{equation}
U=H(K)/D
\end{equation}
where H(K) is the entropy is the key and D is the redundancy of the message.
\label{defn:unicity}
\end{defn}
Definition \ref{defn:perfectsecrecy} leads us to the following theorem:
\begin{thm}
CD-PHY achieves perfect secrecy.
\end{thm}
\begin{proof}
Perfect secrecy requires that without the knowledge of the key, each ciphertext is equally probably to map to any plaintext of that domain. Since the symbols are independent of each other and equally probable to map any of the constellation points, for an M-ary QAM scheme, we have the following:
\begin{equation}
prob(\textbf{C}=C|\textbf{P}=E_K^{-1}(C))= 1/M = prob(\textbf{C}=C)
\end{equation}
which meets the requirements of perfect secrecy. In other words, since the key $K$ is independent of plaintext $P$ and follows uniform distribution, it leads us to:
\begin{equation}
prob(\textbf{P}=P|\textbf{C}=E_K(P))=1/M=prob(\textbf{P}=P)
\end{equation}
More rigorously:\hspace{.5mm}$prob(\textbf{P}=P|\textbf{C}=C)$
\begin{eqnarray}
&=& \frac{prob(\textbf{P}=P, \textbf{C}=C)}{prob(\textbf{C}=C)} \nonumber \\
&=&\frac{prob(\textbf{C}=C|\textbf{P}=P)\hspace{.5mm}prob(\textbf{P}=P)}{\displaystyle\sum\limits_{P'\in \textbf{P}} prob(\textbf{C}=C|\textbf{P}=P')\hspace{.5mm}prob(\textbf{P}=P')} \nonumber \\
&=&\frac{prob(K=C\rightarrow P)\hspace{.5mm}prob(\textbf{P}=P)}{\displaystyle\sum\limits_{P'\in \textbf{P}} prob(K=C\rightarrow P') \hspace{.5mm}prob(\textbf{P}=P')} \nonumber \\
&=&\frac{\frac{1}{M}\hspace{.5mm}prob(\textbf{P}=P)}{\displaystyle\sum\limits_{P'\in \textbf{P}} \frac{1}{M} \hspace{.5mm}prob(\textbf{P}=P')} \nonumber \\
&=&prob(\textbf{P}=P)
\end{eqnarray}
where $K$=$C$$\rightarrow$$P$ refers that key $K$ is a mapping between plaintext $P$ and ciphertext $C$. 
\end{proof}
In addition, according to \textit{perfect cipher keyspace theorem}~\cite{Shannon49}~\footnote{Also known as \textit{Shannon bound}.}, if a cipher is perfect, there must be at least as many keys ($l$) are there are possible messages ($n$). This leads us to the following corollary:
\begin{cor}
Messages in CD-PHY with M-ary QAM scheme should contain less than $n$ symbols such that $M!\geq M^n$ to maintain perfect secrecy.
\end{cor}
Definition \ref{defn:unicity} leads us to the following theorem:
\begin{thm}
The unicity distance of CD-PHY tends to infinity.
\end{thm}
\begin{proof}
For a CD-PHY with M-ary QAM, entropy of the key $H(K)\approx \log M!$. Since, the symbols are independent of each other, the redundancy $D=0$ for the message. So, the unicity distance is $U\approx (\log M!/0) = \infty $. 
\end{proof}
Unicity distance is a theoretical measure of how many ciphertexts are required to determine a unique plaintext. If one has less than unicity distance ciphertext, it is not possible to identify if the deciphering is correct. In fact, when the redundancy approaches to zero, it is hard to attack even simple cipher. For CD-PHY, a unicity distance of infinity means that the eavesdropper won't be able to determine whether the deciphering is correct regardless of the number of the ciphertexts it has in its possession. This is, in fact, a very strong information theoretic guarantee of CD-PHY security.

\subsection{Security by complexity}
\label{complexity}
Now, we model the problem of brute-force key search attack\footnote{Finding the bit sequence to constellation point mapping.} on CD-PHY as a \textit{complete bipartite graph perfect matching} problem and analyze the algorithmic complexity of it.
\begin{defn}
A complete bipartite graph is a bipartite graph where every vertex of the one set is connected to each vertex of the other set. Formally, a complete bipartite graph, $G =(V_1 \cup V_2, E)$, is a bipartite graph such that for any two vertices, $v_1 \in V_1$ and $v2 \in V2$, $v1v2$ is an edge in $G$. 
\end{defn}
From the definition of a complete bipartite graph \cite{Cormen01}, it is straightforward to see the following theorem.
\begin{thm}
The bit sequence to constellation point mapping in CD-PHY is a complete bipartite graph.
\label{thm:bipartite}
\end{thm}
\begin{proof}
A complete bipartite graph partitions the vertices into two sets $|V_1|=p$ and $|V_2|=q$. Now, we can see from Figure \ref{fig:perfectcipher} that each plaintext (bit sequence) on the left side of the graph can be considered a vertex of $V_1$ and each ciphertext (constellation points) on the right can be considered a vertex of $V_2$. Based on the key, it is possible to map every member of $V_1$ to any member of $V_2$. Thus, it constitutes a complete bipartite graph where $|V_1|=|V_2|=\log_2 M$ for an M-ary QAM scheme.
\end{proof}
Now, to explain perfect matching \cite{Plummer92} of the complete bipartite graph, we need the following definition.
\begin{defn}
A matching in a graph is a set of edges without common vertices. In a perfect matching, every vertex of the graph is connected to only one edge of the matching. 
\label{defn:matching}
\end{defn}
The counting version of complete bipartite graph perfect matching problem returns the total number of perfect matching where each edge in the matching connects two unique vertices from $V_1$ and $V_2$. Theorem \ref{thm:bipartite} and Definition \ref{defn:matching} leads us to the following theorem:
\begin{thm}
The brute-force key search attack on CD-PHY is:
\begin{enumerate} 
\item equivalent to counting version of complete bipartite graph perfect matching problem, and
\item in complexity class $\#P$ (Sharp P) complete.
\end{enumerate}
\end{thm}
\begin{proof}
Based on Theorem \ref{thm:bipartite} and Defintion \ref{defn:matching}, proof of part 1 is trivial. The problem of counting the number of perfect matching of a complete bipartite graph can be solved by computing the permanent of the bi-adjacency matrix \cite{Kozen92} of the graph. The permanent of a  matrix $A=n\times n$ is defined as:
%\begin{equation}
%$perm(A)={\sum_{\sigma}}{\prod_{1}^{n}} a_{i,\sigma(i)}$
%\end{equation}
\begin{equation}
perm(A)=\displaystyle\sum\limits_{\sigma}\displaystyle\prod\limits_{1}^{n} a_{i,\sigma(i)}
\label{eq:permanent}
\end{equation}
where $\sigma$ is a permutation over $\{1,2,\ldots,n\}$ . The complexity of computing permanent of a matrix is in complexity class $\#P$ complete, as proved by the seminal work \cite{Valiant79} of Valiant in 1979.
\end{proof}
In general, computing the permanent of a matrix is believed to be harder than its determinant. While one can compute the determinant in polynomial time by Gaussian elimination, the same cannot be used to compute the permanent. Thus, the computational complexity of the brute-force key search attack on CD-PHY also adds to the security of the scheme.  

\subsection{Defense against modulation classification schemes}
The section explains where does CD-PHY stand when the eavesdropper tries to apply some modulation classification techniques such as AMC \cite{Ramkumar09} and DMC \cite{Mobasseri00}.
%\subsubsection{Automatic Modulation Classification (AMC)}
\par AMC is based on cyclic feature detection technique considering the \textit{cyclostationary} property of the modulated signals. It considers the fact that modulated signals in practice have parameters that vary periodically with time. These hidden periodicities are used to classify the modulation techniques. Although, AMC is able to differentiate modulations such as BPSK, QPSK, and QAM based on large amount of training data and supervised learning, it can not identify the shape of the constellation and constellation mapping of symbols to constellation points. Also, for higher order QAM, the complexity of AMC makes it practically infeasible even to classify the modulation. 

\begin{figure*}[htbp]
	\centering
	\begin{minipage}[b]{0.3\linewidth}
		\includegraphics[scale=0.4]{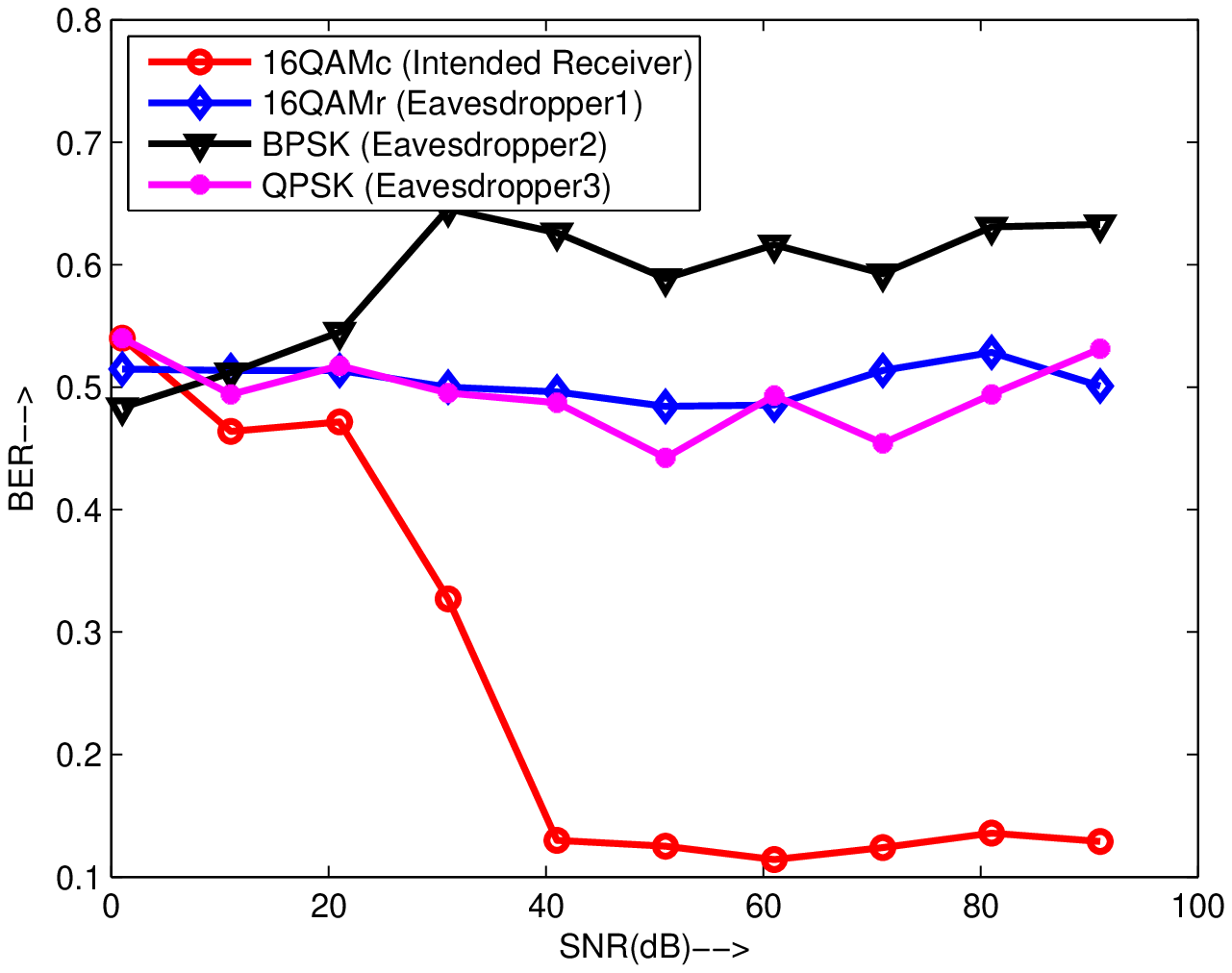}
	\caption{\scriptsize{BER vs SNR for $\alpha=2$ and $d=10m$}}
	\label{fig:scn1}
%\end{figure}
%\begin{figure}[ht]
%	\centering
\end{minipage}
\hspace{0.25cm}
\begin{minipage}[b]{0.3\linewidth}
		\includegraphics[scale=0.4]{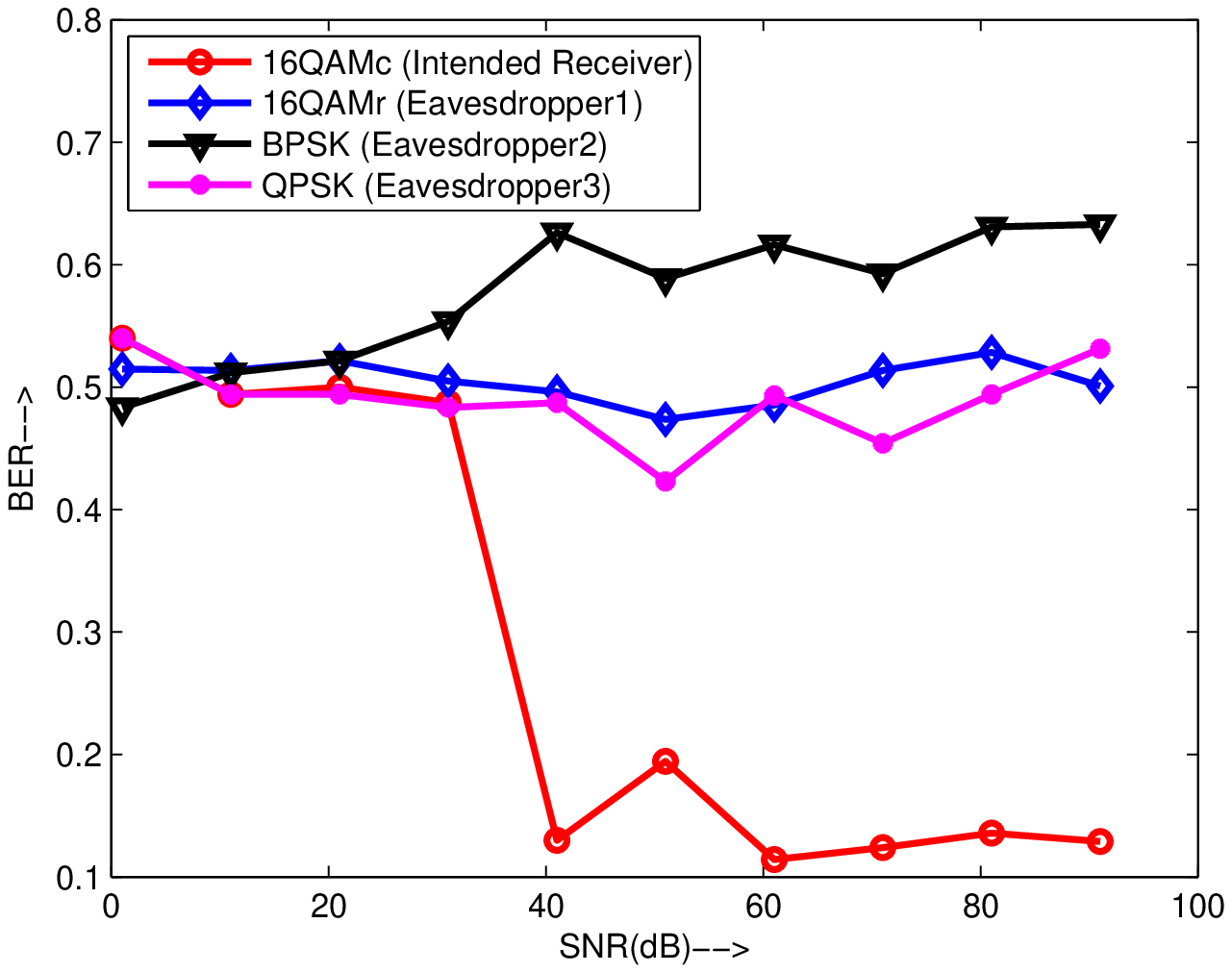}
	\caption{\scriptsize{BER vs SNR for $\alpha=2$ and $d=50m$}}
	\label{fig:scn2}
%\end{figure}
%\begin{figure}[ht]
%	\centering
\end{minipage}
\hspace{0.25cm}
\begin{minipage}[b]{0.3\linewidth}
		\includegraphics[scale=0.4]{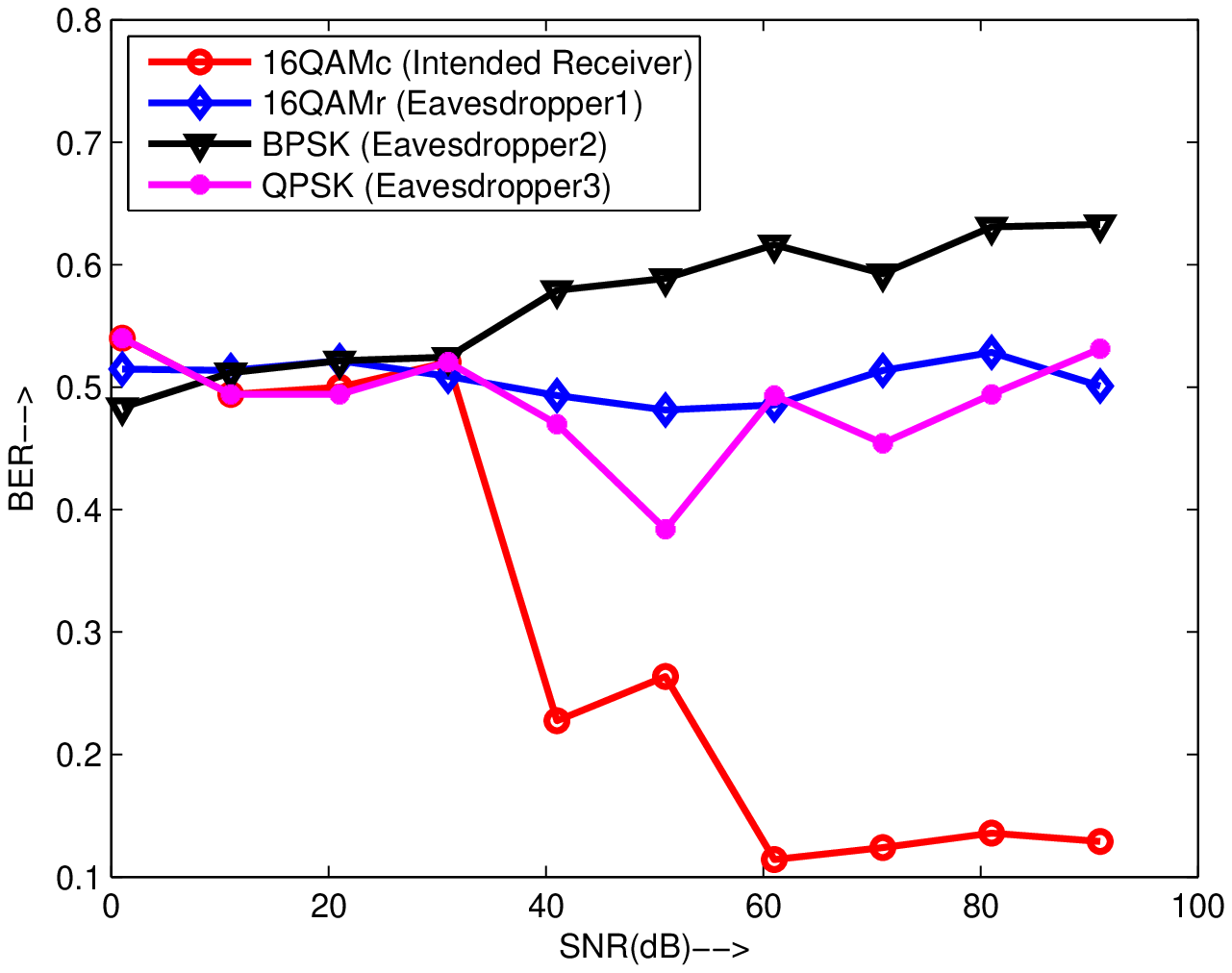}
	\caption{\scriptsize{BER vs SNR for $\alpha=2$ and $d=100m$}}
	\label{fig:scn3}
	\end{minipage}
\end{figure*}

\begin{figure*}[htbp]
\centering
\begin{minipage}[b]{0.3\linewidth}
		\includegraphics[scale=0.4]{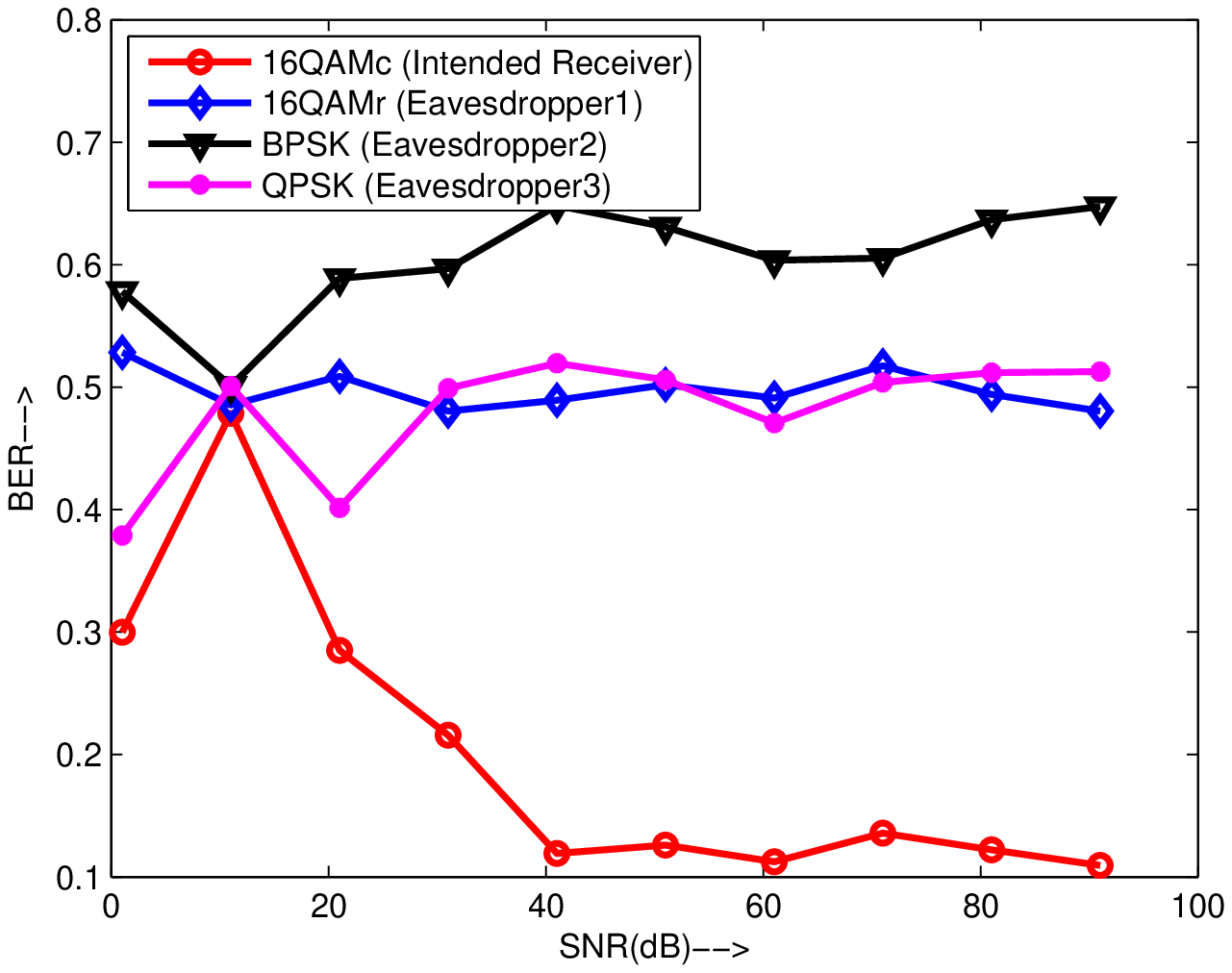}
	\caption{\scriptsize{BER vs SNR for $\alpha=1.4$ and $d=10m$}}
	\label{fig:scn4}
\end{minipage}
\hspace{0.25cm}
%\end{figure}
%\begin{figure}[ht]
\begin{minipage}[b]{0.3\linewidth}
		\includegraphics[scale=0.4]{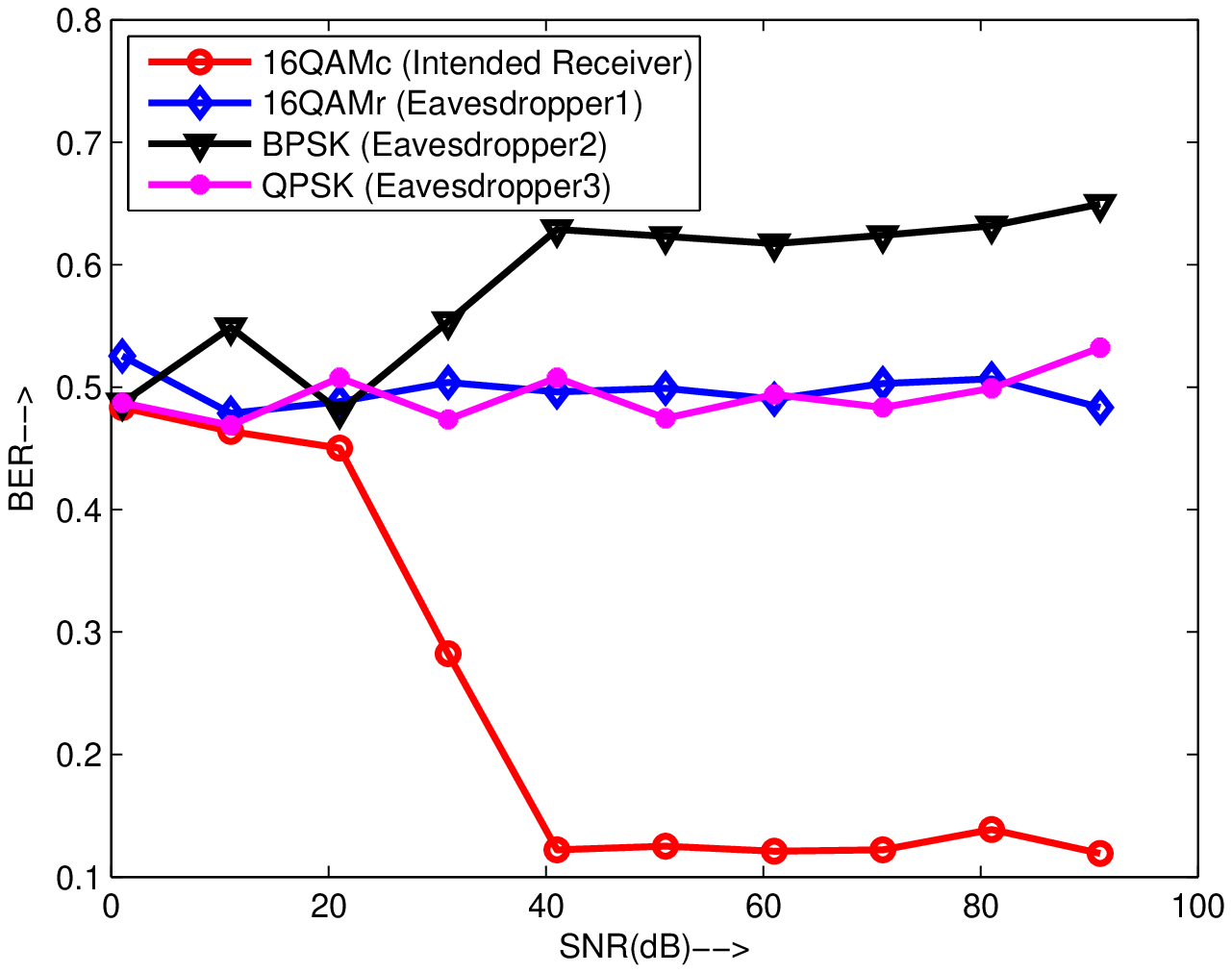}
	\caption{\scriptsize{BER vs SNR for $\alpha=1.4$ and $d=50m$}}
	\label{fig:scn5}
\end{minipage}
\hspace{0.25cm}		
\begin{minipage}[b]{0.3\linewidth}
\includegraphics[scale=0.4]{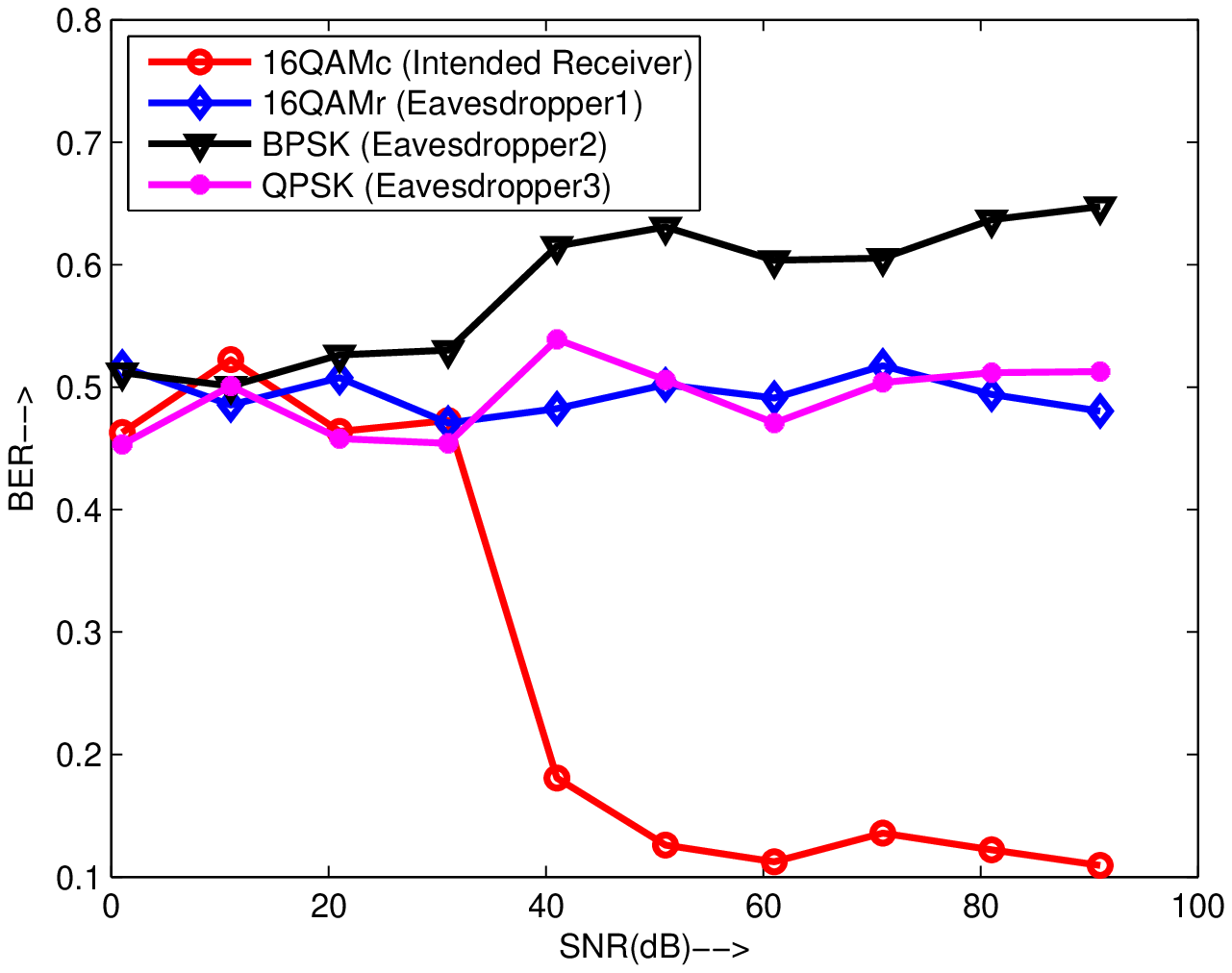}
	\caption{\scriptsize{BER vs SNR for $\alpha=1.4$ and $d=100m$}}
	\label{fig:scn6}
\end{minipage}
\end{figure*} 

\par DMC uses constellation shape as the basis of modulation classification. In this algorithm, the receiver constructs a scatter diagram of the  received noisy symbols in a complex plane and uses fuzzy c-means clustering to recover robust constellation. The modulation type is identified using maximum likelihood (ML) classification with predefined constellation templates. Similar to AMC, digital modulation classification also requires a large amount of training data and supervised learning to identify templates. Thus, although it can identify pre-defined constellation shapes, it is not able to identify constellation mapping from symbols to constellation points.

\par In summary, CD-PHY can withstand existing modulation classification techniques and secure against the attacks exploiting such techniques in practice.
\begin{figure}[ht]
	\centering	
	\includegraphics[scale=0.75]{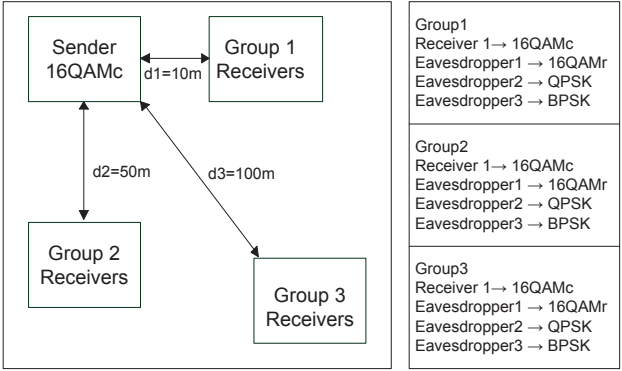}
	\caption{Simulated wireless network scenario. The sender uses 16QAM circular scheme. At different distances, each group has an intended receiver with 16QAM circular scheme and three eavesdroppers each with 16QAM rectangular, QPSK and BPSK scheme.}
	\label{fig:nws}
\end{figure} 
\section{Performance evaluation and simulation results}
\label{eval}
%give and introduction of ber and snr, their relationship and also the relationship with bit error rate.
In this section, we show the impact of CD-PHY on the network performance of the eavesdropper. A very intuitive measure of such performance evaluation is to show how many bits are received in error at different signal and noise power. Typically, when the signal power increases, the receiver is able to decode the bits more accurately leading to a lower bit error rate (BER). In the following experiment, we show that the BER of CD-PHY receiver conforms to this pattern whereas the BER of the eavesdroppers does not decrease even for higher signal power. 
\par The experimental scenario is shown in Figure~\ref{fig:nws}. We designate a CD-PHY sender with 16QAM circular modulation scheme. The receivers are divided into three groups based on their distances from the sender. Group 1, group2 and group 3 are at 10m, 50m and 100m distance, respectively. Each group has an intended CD-PHY receiver with 16QAM circular scheme and three eavesdroppers with 16QAM rectangular, QPSK and BPSK scheme. 
\par We measure the BER at different receivers for different SNRs. Experimental scenarios contain both free space and indoor environments. Figure ~\ref{fig:scn1},~\ref{fig:scn2} and ~\ref{fig:scn3} show the measurements from free space environment. For the CD-PHY receiver, with the increment of SNR, the bit error rate decreases following the usual pattern of wireless communication. However, for eavesdroppers with different schemes, the bit error rate is more than $50\%$ regradless of the increment of SNR. The error rate is the highest in BPSK which is consistent with our analysis in Section \ref{backgroundandobservations}. As the distance increases, BER of BPSK scheme can go as high as $60\%$, resulting in a near to impossible decoding process.
\par Figure ~\ref{fig:scn4},~\ref{fig:scn5} and ~\ref{fig:scn6} show BER vs SNR for indoor environment. The bit error rates of the eavesdroppers are also as high as $50\%$ throughout the measurements for different SNR values. Similar to the free space environment, the distance of the receivers also adversely affect the bit error rate. 
\par Figure~\ref{fig:bercdf} aggregates the BER measurements for different locations of the eavesdropper. The median BER is around $50\%$ and the range is $40\%$ to $60\%$. It shows that in the presence of CD-PHY, the eavesdroppers experience such a high bit error rate that it is almost equivalent of randomly guessing the bits. This is true for both indoor and free space environment and ensures that the eavesdropper can not comprehend the signal when CD-PHY is in action.
\begin{figure}[ht]
	\centering	
	\includegraphics[scale=0.65]{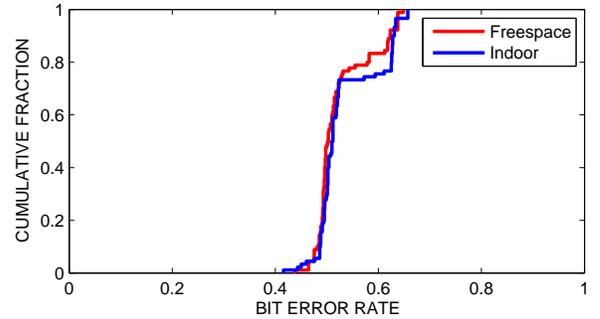}
	\caption{Eavesdropper bit error rate from indoor and free space experiments.}
	\label{fig:bercdf}
\end{figure} 
% add ber for different rata rates.
\section{Conclusion}
CD-PHY is a simple mechanism that introduces channel independent security at the physical layer of wireless communication. We have shown that in the presence of CD-PHY, the eavesdropper has a very low probability of successfully decoding the signal. The scheme achieves Shannon secrecy as a cipher and a brute-force key search attack on CD-PHY falls under complexity class $\#P$ which is believed to be harder than polynomial time algorithms. Our experimental results confirm the theoretical derivations; the bit error rate at the eavesdropper is significantly high and it is practically infeasible to decode the signal which ensures the communication secrecy between the sender and the intended receiver.
\bibliographystyle{IEEEtran}
\bibliography{HusainCDPHY}	
\end{document}